\documentclass[sigconf]{acmart}
\usepackage{hyperref}
\usepackage{url}
\usepackage{mathtools}
\usepackage{amsthm}
\usepackage{multirow}
\usepackage{subcaption}
\usepackage{diagbox}
\usepackage{enumitem}

\AtBeginDocument{%
  }

\setcopyright{acmlicensed}
\copyrightyear{2018}
\acmYear{2018}
\acmDOI{XXXXXXX.XXXXXXX}

\acmConference[Conference acronym 'XX]{Make sure to enter the correct
  conference title from your rights confirmation emai}{June 03--05,
  2018}{Woodstock, NY}
\acmISBN{978-1-4503-XXXX-X/18/06}

\usepackage{amsmath,amsfonts,bm}

\def\figref#1{Figure~\ref{#1}}
\def\tabref#1{Table~\ref{#1}}

\def\secref#1{Section~\ref{#1}}
\def\eqref#1{Equation~(\ref{#1})}
\def\appref#1{appendix~\ref{#1}}

\def\1{\bm{1}}
\def\0{\bm{0}}

\def\vb{{\bm{b}}}

\def\vd{{\bm{d}}}
\def\ve{{\bm{e}}}

\def\vv{{\bm{v}}}
\def\vw{{\bm{w}}}
\def\vx{{\bm{x}}}

\def\evd{{d}}
\def\eve{{e}}

\def\mA{{\bm{A}}}
\def\mB{{\bm{B}}}
\def\mC{{\bm{C}}}

\def\mE{{\bm{E}}}

\def\mI{{\bm{I}}}

\def\mP{{\bm{P}}}
\def\mQ{{\bm{Q}}}
\def\mR{{\bm{R}}}

\def\mU{{\bm{U}}}
\def\mV{{\bm{V}}}

\DeclareMathAlphabet{\mathsfit}{\encodingdefault}{\sfdefault}{m}{sl}
\SetMathAlphabet{\mathsfit}{bold}{\encodingdefault}{\sfdefault}{bx}{n}

\def\emA{{A}}
\def\emB{{B}}

\def\emQ{{Q}}
\def\emR{{R}}

\newcommand{\R}{\mathbb{R}}

\DeclareMathOperator{\Tr}{Tr}
\DeclareMathOperator{\rank}{rank}
\DeclareMathOperator{\diag}{diag}
\DeclareMathOperator{\diagof}{diag\_of}

\newcommand{\ctx}{\mathcal{C}}
\newcommand{\items}{\mathcal{I}}

\newcommand{\Comment}[1]{}

\newcommand{\alex}[1]{\noindent{\textcolor{blue}{\{{\bf Alex:} {\em #1}\}}}}
\newcommand{\michael}[1]{\noindent{\textcolor{red}{\{{\bf Michael:} {\em #1}\}}}}
\newcommand{\naama}[1]{\noindent{\textcolor{magenta}{\{{\bf Naama:} {\em #1}\}}}}
\newcommand{\ariel}[1]{\noindent{\textcolor{orange}{\{{\bf Ariel:} {\em #1}\}}}}
\newcommand{\oren}[1]{\noindent{\textcolor{cyan}{\{{\bf Oren:} {\em #1}\}}}}
\newcommand{\tula}[1]{\noindent{\textcolor{green}{\{{\bf Oren:} {\em #1}\}}}}
\Comment{
\newcommand{\alex}[1]{}
\newcommand{\michael}[1]{}
\newcommand{\naama}[1]{}
\newcommand{\ariel}[1]{}
\newcommand{\oren}[1]{}
\newcommand{\tula}[1]{}
}

\newcommand{\fm}{\mathrm{FM}}

\newtheorem{proposition}{Proposition}
\def\propref#1{Proposition~\ref{#1}}

\newtheorem{identity}{Identity}
\def\identref#1{Identity~\ref{#1}}

\begin{document}

\title{Low Rank Field-Weighted Factorization Machines\\ for Low Latency Item Recommendation}

% Authors must not appear in the submitted version. They should be hidden
% as long as the \iclrfinalcopy macro remains commented out below.
% Non-anonymous submissions will be rejected without review.

\author{Alex Shtoff}
\email{alex.shtoff@yahooinc.com}
\author{Michael Viderman}
\email{viderman@yahooinc.com}
\author{Naama Haramaty-Krasne}
\email{naamah@yahooinc.com}
\affiliation{%
  \institution{Yahoo Research}
  \country{}
  \city{}
}
\author{Oren Somekh}
\email{orens@yahooinc.com}
\author{Ariel Raviv}
\email{arielr@yahooinc.com}
\author{Tularam Ban}
\email{tularam@yahooinc.com}
\affiliation{%
  \institution{Yahoo Research}
  \country{}
  \city{}
}

\renewcommand{\shortauthors}{Shtoff et al.}

% The \author macro works with any number of authors. There are two commands
% used to separate the names and addresses of multiple authors: \And and \AND.
%
% Using \And between authors leaves it to \LaTeX{} to determine where to break
% the lines. Using \AND forces a linebreak at that point. So, if \LaTeX{}
% puts 3 of 4 authors names on the first line, and the last on the second
% line, try using \AND instead of \And before the third author name.

\newcommand{\fix}{\marginpar{FIX}}
\newcommand{\new}{\marginpar{NEW}}

\begin{abstract}
Factorization machine (FM) variants are widely used in recommendation systems that operate under strict throughput and latency requirements, such as online advertising systems. FMs have two prominent strengths. First, is their ability to model pairwise feature interactions while being resilient to data sparsity by learning factorized representations. Second, their computational graphs facilitate fast inference and training. Moreover, when items are ranked as a part of a query for each incoming user, these graphs facilitate computing the portion stemming from the user and context fields only once per query.  Thus, the computational cost for each ranked item is proportional only to the number of fields that vary among the ranked items. Consequently, in terms of inference cost, the number of user or context fields is practically unlimited.

More advanced variants of FMs, such as field-aware and field-weighted FMs, provide better accuracy by learning a representation of field-wise interactions, but require computing all pairwise interaction terms explicitly. In particular, the computational cost during inference is proportional to the square of the number of fields, including user, context, and item. When the number of fields is large, this is prohibitive in systems with strict latency constraints, and imposes a limit on the number of user and context fields for a given computational budget. To mitigate this caveat, heuristic pruning of low intensity field interactions is commonly used to accelerate inference. 

In this work we propose an alternative to the pruning heuristic in field-weighted FMs using a diagonal plus symmetric low-rank decomposition. Our technique reduces the computational cost of inference, by allowing it to be proportional to the number of item fields only. Using a set of experiments on real-world datasets, we show that aggressive rank reduction outperforms similarly aggressive pruning, both in terms of accuracy and item recommendation speed. Beyond computational complexity analysis, we corroborate our claim of faster inference experimentally, both via a synthetic test, and by having deployed our solution to a major online advertising system, where we observed significant ranking latency improvements. We made the code to reproduce the results on public datasets and synthetic tests available at \url{https://anonymous.4open.science/r/pytorch-fm-0EC0}.
\end{abstract}

%%
%% The code below is generated by the tool at http://dl.acm.org/ccs.cfm.
%% Please copy and paste the code instead of the example below.
%%
\begin{CCSXML}
<ccs2012>
   <concept>
       <concept_id>10010147.10010257.10010293.10010309</concept_id>
       <concept_desc>Computing methodologies~Factorization methods</concept_desc>
       <concept_significance>500</concept_significance>
       </concept>
   <concept>
       <concept_id>10002951.10003260.10003272</concept_id>
       <concept_desc>Information systems~Online advertising</concept_desc>
       <concept_significance>300</concept_significance>
       </concept>
   <concept>
       <concept_id>10002950.10003714.10003715.10003719</concept_id>
       <concept_desc>Mathematics of computing~Computations on matrices</concept_desc>
       <concept_significance>500</concept_significance>
       </concept>
 </ccs2012>
\end{CCSXML}

\ccsdesc[500]{Computing methodologies~Factorization methods}
\ccsdesc[300]{Information systems~Online advertising}
\ccsdesc[500]{Mathematics of computing~Computations on matrices}

%%
%% Keywords. The author(s) should pick words that accurately describe
%% the work being presented. Separate the keywords with commas.
\keywords{Recommender systems, Factorization machines, Low rank factorization}

%\received{20 February 2007}
%\received[revised]{12 March 2009}
%\received[accepted]{5 June 2009}

\maketitle

\section{Introduction}
Recommendation systems driven by machine-learned predictive models are widely used throughout the industry for a large variety of applications, from movie recommendation, to ad ranking in online advertising systems. In some applications, recommendation quality is the main objective, where sophisticated and computationally complex deep learning techniques are used to capture the affinity between the users and the recommended items. But other applications require striking an intricate balance between the accuracy of the predictive models, their training and inference speed. For example, real-time bidding systems in programmatic advertising are required to compute a ranking score for a large number of ads in a matter of a few milliseconds, and to train quickly in order to adapt to the ever-changing ad marketplace conditions. Such systems often deploy variants of the celebrated factorization machine (FM) models \citep{fm} in order to overcome data sparsity issues, while excelling at achieving a good balance between prediction accuracy and speed. 

\citet{fm} also showed that while FMs model pairwise feature interactions, whose number is quadratic in the number of features, there is an equivalent formulation of FMs whose computational complexity is \emph{linear} in the number features. This already facilitates fast training. Moreover, when ranking items for a given user in a given context, the user and context features are the same for all items. The equivalent formulation allows caching the user and context computation results once, meaning that the computational cost per item is linear in the number of item features only. Therefore, factorization machines are also extremely efficient for ranking.

Many real-world applications involve large-scale multi-field data (e.g., gender and item category). However, FMs have a limited predictive accuracy, since they fail to capture the fact that the same feature can behave differently when interacting with features from different fields. To resolve this issue, many variants that incorporate field information have been proposed in recent years, including field-aware \citep{ffm}, field-weighted \citep{fwfm}, and field-embedded \citep{fefm,fmfm} factorization machines. Unfortunately, none of these variants admit an equivalent formulation whose computational complexity is linear in the number of features. This poses a challenge when building cost-effective, large-scale real-time recommendation systems, as the additional gain from incorporating field information comes at the cost of additional computing power.

The focus of this work is the \emph{field-weighted factorization machine} (FwFM) variant \cite{fwfm}, under the large-scale and low inference latency regime. It is similar to a regular factorization machine, with an additional symmetric matrix of parameters modeling the strength of pairwise field interactions. It is attractive in practice due to the fact that in terms of memory consumption and the number of parameters, it is on par with a regular FM. In addition, it is less prone to over-fitting, as pointed out by \citet{ffm} and \citet{fwfm}, and admits a simple heuristic for reducing the computational complexity at the inference stage by pruning low-magnitude field interactions. However, pruning has some cost in terms of accuracy. In this work, we devise a different way to significantly reduce the computational cost of field-weighted factorization machines by decomposition the matrix of pairwise field interactions. Motivated by the visualization of the field interaction matrices in \cite{fwfm}, that resembles a block-like structure due to field groups exhibiting similar interaction behavior, we employ the well-known diagonal plus low-rank (DPLR) matrix decomposition, to produce an FwFM variant we call DPLR-FwFM. We show that our idea facilitates utilizing the additional accuracy provided by incorporating field information, while benefiting from a per-item computational cost that is linear in the number of item fields, albeit slower than regular FM by a factor proportional to the rank of the low-rank part of the decomposition.

We evaluate our technique on public datasets, and on proprietary data from a large-scale online advertising system of a major company. Our results demonstrate that, in practice, DPLR-FwFM with extremely low ranks outperform aggressively pruned FwFM models both in terms of latency and accuracy. \Comment{Moreover, our experiments show that a "post-hoc" approach of creating a DPLR-FwFM model after training a regular FwFM by decomposing the field interaction matrix is worse than both a regular FwFM, and a pruned FwFM of comparable computational cost. \ariel{too detailed for intro}\alex{Is this better?}\ariel{I would drop the entire sentence}}. Finally, we deploy our solution in an online advertising system of a major company, and show that the latency incurred by a prediction module based on our DPLR-FwFM model is better than that of the module currently used in production and is based on a pruned FwFM.

To summarize, the main contributions of our paper are:
\begin{enumerate}
    \item Reformulate FwFM models to make their computational cost on par, or higher by a small constant factor of our choice, compared to regular FMs, and significantly cheaper than FwFMs. This, while benefiting from the higher accuracy provided by incorporating field information.
    \item Show that the accuracy of the obtained models is on par or higher than pruned FwFM models on public and proprietary data-sets.
    \item Demonstrate that, in practice, in a real-world online advertising system, our approach can significantly reduce the computational costs, thus achieving lower latency with a given computational budget, or reducing the computational power required to achieve a given latency. 
\end{enumerate}
\section{Related work}

Recommendation technologies such as \textit{Collaborative Filtering} (CF) \citep{goldberg1992using}, help users discover new items based on past preferences. \textit{Matrix Factorization} (MF) is a leading approach in CF, addressing data expansion and sparsity by using a latent factor model \citep{koren2009matrix}. Such systems find applications in various domains, including movie recommendation \citep{bell2007lessons}, music recommendation \citep{aizenberg2012build}, ad matching \citep{aharon2013off} and more. 
\textit{Factorization machine} (FM) variants excel in benchmark tabular classification tasks, particularly when dealing with a large number of interactions and requiring fast predictions. This family includes several members as the original FM \citep{rendle2010factorization}, \textit{field-aware factorization machine} (FFM) \citep{juan2016field}, \textit{field weighted factorization machine} (FwFM) \citep{pan2018field} and the \textit{field-matrixed factorization machine} (FmFM) \citep{sun2021fm2}.
For handling a wider range of features with non-linear interactions, \cite{cheng2016wide} presented \textit{Wide \& Deep} learning approach, which trains a network that combines a linear model and a deep neural network. Their work was followed by a wave of deep learning techniques aimed at improving prediction accuracy, while also being sensitive to low-latency efficiency. An essential application that motivated such works, and also the focus of this work, is CTR prediction for online advertising. It has been the topic of many techniques as DeepFM \citep{guo2017deepfm}, xDeepFM \citep{lian2018xdeepfm}, DCN \citep{wang2017deep}, AutoInt \citep{song2019autoint}, DeepLight \citep{deng2021deeplight} and more. While these methods offer improved predictions, their high latency and long training time cause many applications to still prefer FM-based approaches.

\Comment{\oren{mention that we are focused on application that are sensitive to low latency efficient such as CTR prediction for online advertising. In the next sentences focus on approaches that combine FM and NN} }
\Comment{
Their work was followed by a surge of techniques that combine FM and neural network to improve predictions through deep learning, such as DeepFM \cite{guo2017deepfm}, xDeepFM \cite{lian2018xdeepfm}, NFM \cite{he2017neural}, DCN \cite{wang2017deep}, AutoInt \cite{song2019autoint}, DeepLight \cite{deng2021deeplight} and more. While these methods offer improved predictions, their high latency and long training time cause many applications to still prefer FM-based approaches.
}

Low-rank factorizations have emerged as a versatile and powerful tool in \textit{machine learning} (ML) and \textit{artificial intelligence} (AI), finding applications across a spectrum of domains. For instance, low-rank matrix-factorization in the context of recommendation systems \citep{koren2009matrix, rendle2010factorization}, in natural language processing for topic modeling, and dimensionality reduction \citep{blei2003latent}, and more recently, for large-scale pre-training and fine-tuning of models on diverse natural language understanding tasks \citep{hu2021lora}.
In cases where the involved matrices are full-ranked and low-rank factorizations are less effective, \textit{diagonal plus low-rank factorization} (DPLR) may be considered (see the work of \cite{saunderson2012diagonal} and references therein).
DPLR involves approximating a given matrix, often a large, and high-dimensional one, by decomposing it into the sum of two matrices: a diagonal matrix and a low-rank matrix. In statistics, DPLR is used to approximate high-dimensional covariance matrices of multivariate normal distributions \citep{liutkus2017diagonal,ong2018gaussian}. In ML, DPLR is used for enhancing the efficiency of models (see the work of \cite{bingham2019pyro} for an ML library that uses DPLR covariance matrices). For instance, DPLR-based methods have been employed in deep learning architectures for more efficient and accurate models \citep{zhao2016low,tomczak2020efficient,mishkin2018slang}.  Inspired by an observation of \cite{convex_fm}, that the regular FM field-interaction matrix equals a rank-one matrix minus the identity matrix, we apply DPLR to FwFM and improve its training and inference efficiency.

One FM-like approach that resembles our work is of \cite{almagor2022you}, where the authors apply a low-rank factorization to a generalized FM family representation (that includes FM, FFM, FwFM and FmFM as special cases). \Comment{Then, they leveraged the known connection between FM to polynomial neural networks \citep{blondel2016polynomial} to propose a neural network-based framework referred to as \textit{field-wise factorized neural-network} (FiFa). The proposed framework outperforms other models such as FmFM for CTR prediction over tabular data. Furthermore, it exhibits improved performance and inference times compared to deeper methods, such as SAINT \citep{somepalli2021saint}.}At first glance, their method seems to have much in common with our work, as both works use low-rank factorization to represent field interactions. However, there are two key differences. First, \cite{almagor2022you} takes a modeling approach to show that a family of FMs can be implemented as standard feed-forward networks. Our paper takes a different direction, that of  recommendation systems: we focus on using a low-rank factorization to represent the separation into context and item fields, to improve item recommendation accuracy with FwFM models under strict latency constraints. Second, in contrast to \cite{almagor2022you}, our work adds a diagonal component to the low-rank factorization, which plays a key role in the ability to be efficient for real-time systems, as shown in \secref{sec:dplr_math_foundation}.

\section{Background and problem formulation}
In this section, we provide a detailed explanation of the computational efficiency properties of FMs in recommending items, and the efficiency drop introduced by switching to field-aware variants. We then formulate the problem of reducing the gap between the two.

\subsection{Efficiency of factorization machines}
Given a feature vector $\vx\in \R^n$, the FM model proposed by \citet{fm} computes
\[
\Phi_{\fm}(\vx; b_0, \vb, \vw_1, \dots, \vw_n) = b_0 + \langle \vb, \vx \rangle + \sum_{i=1}^n \sum_{j=i+1}^n \langle x_i \vw_i, x_j \vw_j \rangle,
\]
where $b_0 \in \R$, $\vb \in \R^n$, and $\vw_1, \dots, \vw_n \in \R^k$ are trainable parameters. Overall, the model has $1 + n + n k$ trainable parameters, and the direct formula above for $\Phi_{\fm}$ can be computed in $O(n^2 k)$ time. But as \citet{fm} pointed out, the pairwise interaction terms can be re-written as
\begin{equation}\label{eq:fm_feat_fast}
\sum_{i=1}^n \sum_{j=i+1}^n \langle x_i \vw_i, x_j \vw_j \rangle = \frac{1}{2} \left(
 \Bigl \| \sum_{i=1}^n x_i \vw_i \Bigr \|^2 - \sum_{i=1}^n \Bigl \| x_i \vw_i \Bigr \|^2
\right).
\end{equation}
This dramatically reduces the time complexity to $O(n k)$.

In practice, the feature vector encodes a row in a tabular data-set of past interactions between users and items, whose columns, often named \emph{fields}, contain categorical features. Some columns describe the context (including the user), whereas others describe the item. Thus, $\vx$ formally consists of field-wise one-hot encodings, for example:
\[
\vx = (
    \underbrace{
        \underbrace{0, 1, 0, 0, 0,}_{\text{field 1}} 
        \underbrace{1, 0, 0, 0,}_{\text{field 2}}
        \dots
        \underbrace{0, 0, 1, 0, 0, 0,}_{\text{field $m_c$}}
    }_{\text{context fields}}
    \underbrace{
        \underbrace{0, 0, 1, 0,}_{\text{field $m_c + 1$}}
        \dots
        \underbrace{0, 0, 0, 0, 1}_{\text{field $m$}}
    }_{\text{item fields}}
)^T,
\]
Therefore, when computing $\Phi_{\fm}$ we can only use the $m$ nonzero entries of $\vx$ corresponding to the $m$ fields. We note that there may be fields having multiple values, such as a list of movie genres, but we defer their treatment to the discussion of FwFM models in the sequel.

Formally, for a given feature vector $\vx$, let $\ell_1, \dots, \ell_m$ denote the nonzero indices, and define $\vv_i \equiv \vw_{\ell_i}$. Hence, the formula in \eqref{eq:fm_feat_fast} can be reformulated by summing over the fields, rather than the features. Moreover, we can split the summation over all fields to a summation over the context fields $\ctx = \{1, \dots, m_c\}$ and the item fields $\items = \{m_c + 1, \dots, m\}$, and obtain:
\begin{subequations}
\begin{align}
\sum_{i=1}^n \sum_{j=i+1}^n &\langle x_i \vw_i,  x_j \vw_j \rangle  \\
 &= \sum_{i=1}^m \sum_{j=i+1}^m \langle \vv_i, \vv_j \rangle \label{eq:fm_fieldwise} \\
 &= \frac{1}{2} \left(
 \Bigl \| \sum_{i=1}^m \vv_i \Bigr \|^2 - \sum_{i=1}^m \Bigl \| \vv_i \Bigr \|^2
 \right) \\
 &= \frac{1}{2} \left(
 \Bigl \| \sum_{i \in \ctx} \vv_i + \sum_{i \in \items} \vv_i \Bigr \|^2 - \sum_{i \in \ctx} \Bigl \| \vv_i \Bigr \|^2 - \sum_{i\in \items} \Bigl \| \vv_i \Bigr \|^2
 \right) \label{eq:fm_filedwise_sep}
\end{align}
\end{subequations}
By \eqref{eq:fm_filedwise_sep}, we see that when ranking a large number of items for a given context, the sums over $\ctx$ can be computed only \emph{once}. For each item, the computational complexity is $O(|\items| k)$, namely, and it depends largely on the number of \emph{item fields}. Thus, in practice, we can use a model with a large number of user or context features to achieve a high degree of personalization, without incurring a significant computational cost during item ranking.

\subsection{Inefficiency of field-weighted factorization machines}\label{sec:inefficiency}
Field-weighted factorization machines (FwFM) \citep{fwfm} model the varying behavior of a feature belonging to some field when interacting with features from different fields in the form of a trainable symmetric field interaction matrix $\mR \in \R^{m\times m}$. Its components $\emR_{i,j}$ model the \emph{intensity} of the interaction between field $i$ and field $j$. The model's output is
\begin{multline*}
\Phi_{\mathrm{FwFM}}(\vx; b_0, \vb, \vw_1, \dots, \vw_n, \mR) = \\ b_0 + \langle \vb, \vx \rangle + \sum_{i=1}^n \sum_{j=i+1}^n \langle x_i \vw_i, x_j \vw_j \rangle \emR_{f_i,f_j},
\end{multline*}
where $f_i$ is the field corresponding to feature $i$. \citet{fwfm} demonstrate that this model family achieves a significantly higher accuracy compared to a regular FM, and is comparable with other field-aware variants. One attractive property of this variant is its number of parameters and memory requirements - it has only $\frac{m(m-1)}{2}$ additional parameters compares to a regular FM comprising the above-diagonal entries of the field interaction matrix $\mR$.  

Similarly to \eqref{eq:fm_fieldwise}, under the one-hot encoding assumption, given an input $\vx$, the pairwise interaction term can be written in terms of the field vectors $\vv_1, \dots, \vv_m$ that are a subset of the embedding vectors $\vw_1, \dots, \vw_n$ that correspond to the $m$ nonzero entries of $\vx$:
\begin{equation}\label{eq:fieldwise_fwfm}
    \sum_{i=1}^n \sum_{j=i+1}^n \langle x_i \vw_i, x_j \vw_j \rangle \emR_{f_i,f_j} = \sum_{i=1}^m \sum_{j=i+1}^m \langle \vv_i, \vv_j \rangle \emR_{i, j}
\end{equation}
However, the time complexity of computing its output $\Phi_{\mathrm{FwFM}}$ is dominated by the $O(m^2 k)$ complexity of computing the pairwise term in \eqref{eq:fieldwise_fwfm}, which is quadratic in the total number of fields. Compared to the $O(|\items| k)$ per item complexity of a regular FM, which is linear in the number of item fields, FwFMs pose a serious challenge for applications where inference speed is critical. While the quadratic complexity property is shared by many other field-aware variants, such as \citet{ffm, fefm, fmfm}, in this work, we focus on addressing this challenge for FwFM models.

The one-hot encoding assumption does not cover the case when a field has multiple values, such as a list of movie genres. In this case, we assume that multiple components of a field's encoding, that correspond to multiple values, may be nonzero. We handle this case by assuming that a field does not interact with itself, i.e. $R_{f,f} = 0$ for any field $f$. A direct consequence is that the pairwise interaction term can be written in terms of field vectors, as in \eqref{eq:fieldwise_fwfm}, but each vector $v_i$ may be a weighted sum of the corresponding feature embedding vectors. For example, a movie with 3 genres may be encoded by placing $\frac{1}{3}$ in the corresponding components of $\vx$, which will result in an average of the genre embedding vectors.

\subsection{Problem formulation}
A typical approach used in practice to speed up inference is pruning the matrix $\mR$ by zeroing-out entries whose magnitude is below a threshold, as suggested by \citet{fwfm} and \citet{fmfm}. However, as we show in \secref{sec:experiments}, aggressive pruning reduces the accuracy of the model. In this work, we aim to devise a method to achieve inference speed that is proportional to the number of \emph{item fields}, and is only $\rho$ times slower than a regular FM, where $0<\rho\ll m$ is a configurable factor. This, while retaining an accuracy that is comparable with that of a regular FwFM model, and significantly better than pruned model with a similar number of parameters.
\section{Low-rank field-weighted factorization machines}
In this section we reformulate the pairwise field interaction term in \eqref{eq:fieldwise_fwfm} in an equivalent, but significantly more efficient manner. Throughout this section, we assume that the field vectors $\vv_1, \dots, \vv_m \in \R^{k}$ are embedded into the \emph{rows} of the matrix $\mV \in \R^{m \times k}$:
\begin{equation}\label{eq:vector_in_rows}
\mV = \begin{bmatrix}
    \, \textrm{---} \, \vv_1 \, \textrm{---}\,  \\
    \, \textrm{---} \, \vv_2 \, \textrm{---}\,  \\
    \vdots \\
    \, \textrm{---} \, \vv_m \, \textrm{---}\, 
\end{bmatrix}
\end{equation}
Moreover, since \eqref{eq:fieldwise_fwfm} uses only the upper triangular part of $\mR$, we may choose the remaining entries arbitrarily to our convenience, and throughout this section we assume that $\mR$ is symmetric with a zero diagonal. Under this assumption, \eqref{eq:fieldwise_fwfm} can be re-written as
\begin{equation}\label{eq:fieldwise_fwfm_full}
\sum_{i=1}^m \sum_{j=i+1}^m \langle \vv_i, \vv_j \rangle \emR_{i,j} = \frac{1}{2} \sum_{i=1}^m \sum_{j=1}^m \langle \vv_i, \vv_j \rangle \emR_{i,j}.
\end{equation}
In the sequel, we describe our method to efficiently compute the double sum in the right-hand side of \eqref{eq:fieldwise_fwfm_full}. We first review the mathematical background, and then the complete method.

\subsection{Mathematical foundation}\label{sec:dplr_math_foundation}
In this section we present a technical result that provides the motivation for factorizing the matrix $\mR$ to obtain an efficient algorithm. Then, we present an example explaining why the widely used \emph{low-rank factorization} lacks expressive power, and use it to motivate a \emph{diagonal plus low-rank factorization}.

The following identity let us reformulate the pairwise interaction formula as the one in \eqref{eq:fieldwise_fwfm_full} in a matrix form, that is inspired by the reformulation in \citet{convex_fm} for regular FMs.
\begin{identity}\label{ident:pairwise_as_trace}
Let $\mA \in \R^{m \times k}$. Then, for any matrix $\mQ \in \R^{m \times m}$ we have
\begin{equation}\label{eq:pairwise_trace_identity}
    \sum_{i=1}^m \sum_{j=1}^m \langle \mA_{i,:}, \mA_{j,:} \rangle \emQ_{i,j} = \Tr(\mA^T \mQ \mA)
\end{equation}
\end{identity}
\begin{proof}
Recall the definition of the Frobenius inner product of matrices:
\[
\langle \mA, \mB \rangle = \sum_{i=1}^m \sum_{j=1}^n \emA_{i,j} \emB_{i,j} = \Tr(\mA^T \mB),
\]
and the circular shift invariance property of the trace operator \citep[Section 1.1]{matrixcookbook}:
\[
\Tr(\mA \mB \mC) = \Tr(\mB \mC \mA)
\]

For any $1 \leq i,j \leq m$ it holds that $\langle \mA_{i,:}, \mA_{j,:} \rangle = (\mA \mA^T)_{i,j}$, and thus    
\begin{multline*}
\sum_{i=1}^m \sum_{j=1}^m \langle \mA_{i,:}, \mA_{j,:} \rangle  \emQ_{i,j} \\
 = \sum_{i=1}^m \sum_{j=1}^m \left(\mA \mA^T\right)_{i,j} \emQ_{i,j} = \langle \mA \mA^T, \mQ \rangle = \Tr(\mA\mA^T \mQ),
\end{multline*}
where the last two equalities follow from the definition of the Frobenius inner product, and the fact that $\mA\mA^T$ is symmetric. Finally, the circular shift invariance property implies that $\Tr(\mA \mA^T \mQ) = \Tr(\mA^T \mQ \mA)$, completing the proof.
\end{proof}

Since $\mR$ is a real symmetrical square matrix its eigenvalue decomposition is composed of real eigenvalues. If the eigenvalues decay quickly, we can use an approximation obtained by using the $\rho$ largest magnitude eigenvalues, for some $0<\rho\ll m$, in the form $\mR = \mU^T \diag(\ve) \mU$, where $\mU\in \R^{\rho\times m}$ and $\ve \in \R^{\rho}$. 
\Comment{Thus, invoking \identref{ident:pairwise_as_trace} twice we could devise a more efficient algorithm for computing the pairwise interactions, since
\begin{align*}
\sum_{i=1}^m \sum_{j=1}^m \langle \vv_i, \vv_j \rangle \emR_{i,j} = \Tr(\underbrace{(\mU \mV)^T}_{\mP^T} \diag(\ve) \underbrace{(\mU \mV)}_{\mP}) = \sum_{i=1}^\rho \eve_i \| \mP_{i,:} \|^2
\end{align*}
Since $\mU$ has only $\rho$ rows, computing the product $\mU \mV$ costs only $O(\rho m k)$, and the computation becomes significantly more efficient. }

However, a low-rank decomposition lacks the expressive power we need. To see why, consider a regular FM, which is equivalent to an FwFM where all field interactions are 1, meaning that the field-interaction matrix is:
\[
\mR_{\fm} = \begin{pmatrix}
    0 & 1 & \dots & 1 \\
    1 & 0 & \dots & 1 \\
    \vdots & \vdots & \ddots & \vdots \\
    1 & 1 & \dots & 0
\end{pmatrix}
\]
We have $\rank(\mR_\fm) = m$, and its eigenvalues are $\{m-1, -1, \dots, -1 \}$, and hence a low-rank approximation within a reasonable accuracy is impossible. Since a low-rank decomposition does not even have enough expressive to model a regular FM, an FwFM is clearly out of reach. The obstacle stems from the zero diagonal of $\mR$, which is an ``anomaly'' in the matrix. But as \citet{convex_fm} pointed out, this anomaly can be easily handled using a diagonal matrix:
\begin{equation}\label{eq:fm_diag_lowrank}
    \mR_\fm = \1 \1^T - \mI,
\end{equation}
where $\1$ is a column vector whose components are all 1. In other words, the matrix $\mR_\fm$ can be decomposed into a sum of a diagonal matrix and a low-rank matrix. Motivated by the above, we use a diagonal plus low-rank (DPLR) decomposition to facilitate fast inference, as is shown in the proposition below. The low-rank part resembles the eigenvalue decomposition, to be able to model arbitrary symmetric matrices, including indefinite ones.
\begin{proposition}\label{prop:pairwise_low_rank}
Let $\mV$ be as defined in \eqref{eq:vector_in_rows}, let $\mR \in \R^{m \times m}$ be a symmetric matrix, and suppose that $\mU \in \R^{\rho \times m}$, $\ve \in \R^\rho$, and $\vd \in \R^m$ are such that
\begin{equation}\label{eq:dplr}
    \mR = \mU^T \diag(\ve) \mU + \diag(\vd).
\end{equation}
Define $\mP = \mU \mV$. Then,
\begin{equation}\label{eq:efficient_pairwise}
\sum_{i=1}^m \sum_{j=1}^m \langle \vv_i, \vv_j \rangle \emR_{i,j} = \sum_{i=1}^m \evd_i \| \vv_i \|^2 + \sum_{i=1}^\rho \eve_i \| \mP_{i,:} \|^2
\end{equation}
\end{proposition}
\begin{proof}
Invoking \identref{ident:pairwise_as_trace} with $\mA = \mV$ and $\mQ = \mR$, and then using the decomposition in \eqref{eq:dplr} we have
\begin{align*}
\sum_{i=1}^m \sum_{j=1}^m \langle \vv_i, \vv_j \rangle \emR_{i,j} 
 &= \Tr(\mV^T \mR \mV) \\
 &= \Tr\left(\mV^T(\diag(\vd) + \mU^T \cdot \diag(\ve) \cdot \mU) \mV\right) \\
 &= \Tr(\mV^T \diag(\vd) \mV) + \Tr(\underbrace{(\mU \mV)^T}_{\mP^T} \diag(\ve) \underbrace{(\mU \mV)}_{\mP}).
\end{align*}
Invoking \identref{ident:pairwise_as_trace} with $\mA = \mV$ and $\mQ = \diag(\vd)$, we obtain
\[
\Tr(\mV^T \diag(\vd) \mV) = \sum_{i=1}^m \evd_i  \| \vv_i  \|^2,
\]
and again with $\mA = \mP$ and $\mQ = \diag(\ve)$ we obtain
\[
\Tr(\mP^T \diag(\ve) \mP) = \sum_{i=1}^\rho \eve_i  \| \mP_{i,:}   \|^2.
\]
Therefore,
\[
\sum_{i=1}^m \sum_{j=1}^m \langle \vv_i, \vv_j \rangle \emR_{i,j}  = \sum_{i=1}^m d_i  \| \vv_i  \|^2 + \sum_{i=1}^\rho e_i  \| \mP_{i,:}   \|^2
\]
\end{proof}
Now we are ready to present our solution based on \propref{prop:pairwise_low_rank}.

\subsection{The modified model and fast inference algorithm}
Our solution consists of a modification of the FwFM model, and an algorithm that achieves inference during ranking in $O(\rho |\items| k)$ per item.

\subsubsection{The diagonal plus low-rank FwFM model}
Instead of learning the field interaction matrix $\mR$ directly, we learn its decomposition in the diagonal plus low-rank form. Since $\diag(\mR) = \0$, the diagonal component is fully determined by the low-rank decomposition. Formally, we replace the learned parameter $\mR$, with the learned parameters $\mU \in \R^{\rho \times m}$ and  $\ve \in \R^{\rho \times m}$, where $\rho$ is a hyper-parameter. The matrix $\mR$ is formally defined as
\begin{equation}\label{eq:dplr_model_decomposition}
\mR = \mU^T \diag(\ve) \mU + \diag(\vd),
\end{equation}
where 
\[
\vd \equiv -\diagof(\mU^T \diag(\ve) \mU).
\]
In other words, we ensure that $\mR$ is defined to be a symmetric matrix with a zero diagonal by forming a symmetric eigenvalue-like decomposition, and subtracting the diagonal of the resulting matrix. The formal definition of $\mR$ allows us to apply \propref{prop:pairwise_low_rank}. This means that, we do not  need to compute the matrix $\mR$ itself at any stage.  Instead, we replace the FwFM pairwise interaction term in \eqref{eq:fieldwise_fwfm}, with the outcome of the proposition in \eqref{eq:efficient_pairwise}. Namely, we compute $\mP = \mU \mV$ in $O(\rho m k)$ time, and then compute $\frac{1}{2} \left( \sum_{i=1}^m \evd_i \| \vv_i \|^2 + \sum_{i=1}^\rho \eve_i \| \mP_{i,:} \|^2 \right)$ in $O(\rho k + m k)$ time.

We call the resulting model a diagonal plus low-rank field weighted factorization machine, or DPLR-FwFM. Although this is not our main objective, a nice byproduct is that training  also becomes slightly faster, since the above two steps cost $O(\rho m k)$ time instead the na\"ive $O(m^2 k)$. This is important for some online advertising systems that require deploying a ``fresh'' model as soon as possible, to adapt to the quickly evolving marketplace conditions \citep{debt,trends,zhang2020retrain}. However, for inference during \emph{ranking} we need to take the result of the proposition one step further.

\subsubsection{Fast inference for item ranking}
Observe that the field vectors embedded into the rows of $\mV$ can be decomposed into context and item vectors:
\[
\mV = \begin{bmatrix}
      \mV_\ctx \\
      \vphantom{m} \\
      \mV_\items
\end{bmatrix}
\]
The corresponding columns of $\mU$ can be decomposed similarly:
\[
\mU = \begin{bmatrix}
    \mU_\ctx & \mU_\items
\end{bmatrix}
\]
Consequently, matrix $\mP$ can be written as:
\[
    \mP = \mU_\ctx \mV_\ctx + \mU_\items \mV_\items,
\]
The product $\mU_\ctx \mV_\ctx$ can be computed only once when ranking with a given context. The sum $\sum_{i=1}^m \evd_i \| \vv_i \|^2$ can also be similarly decomposed as
\[
  \sum_{i=1}^m \evd_i \| \vv_i \|^2 
    = \sum_{i\in \ctx} \evd_i \| \vv_i \|^2 + \sum_{i \in \items} \evd_{i} \| \vv_{i} \|^2,
\]
where $\sum_{i\in \ctx} \evd_i \| \vv_i \|^2$ can be computed once when ranking with a given context.
To summarize, during ranking, the pairwise interaction term of each item during ranking using DPLR-FwFM can be computed by the following algorithm: \\
\fbox{
\begin{minipage}[p]{.95\linewidth}
    \begin{center}
    \textbf{Algorithm 1} \\
    DPLR FwFM pairwise interaction inference with cached context
    \end{center}
    \noindent\hrulefill\\
    \textbf{Input}: 
    \begin{itemize}
        \item The field matrix $\mV$, partitioned into $\mV_\ctx$, $\mV_\items$
        \item $\mU \in \R^{\rho \times m}$, $\vd \in \R^m$, $\ve \in \R^\rho$, that form the decomposition $\mR = \mU^T \diag(\ve) \mU + \diag(\vd)$, with $\mU$ partitioned into $\mU_\ctx$, $\mU_\items$
    \end{itemize}
    \textbf{Output}: The pairwise interactions $\sum_{i=1}^m \sum_{j=1}^m \langle \vv_i, \vv_j \rangle \emR_{i,j}$ \\
    \textbf{Steps}:
    \begin{enumerate}
        \item Once per context, compute:
        \begin{enumerate}
            \item Compute $\mP_\ctx = \mV_\ctx \mU_\ctx$ 
            \item Compute $s_\ctx = \sum_{i \in \ctx} \evd_i \| \vv_i \|^2$ 
        \end{enumerate}
        \item Compute $\mP = \mP_c + \mU_\items \mV_\items$ 
        \item Return $s_\ctx + \sum_{i\in \items} \evd_{i} \| \vv_{i} \|^2 + \sum_{i=1}^\rho \eve_i \| \mP_{i,:}  \|^2$
    \end{enumerate}    
\end{minipage}
}

\begin{table*}[h]
    \caption{Results for public data-sets. A lower LogLoss or MSE, and a higher AUC are better. The pruned sparsity column shows the percentage of entries in the pruned model equivalent to a DPLR model of the given rank, in terms of the number of parameters. The FM/FwFM columns show results without pruning or rank reduction, for reference, and are replicated across ranks. The next columns show results with pruning and rank reduction. Finally, the last column shows the improvement percentage of the DPLR over the equivalently pruned model.}
    \label{tbl:public_ds_results}
    \centering   
    \footnotesize
    \begin{tabular}{lll|rrrrp{.08\linewidth}|rrrrp{.08\linewidth}}
    \toprule
     \multirow[t]{2}{*}{Dataset (Metric)} & \multirow[t]{2}{*}{Rank} & \multirow{2}{*}{\shortstack[l]{Pruned \\ sparsity}} & \multicolumn{5}{c}{$k=8$} & \multicolumn{5}{c}{$k=16$} \\
     \cline{4-8} \cline{9-13}
      & &  & FM & FwFM & DPLR & Pruned & {DPLR vs \par Pruned (\%)} & FM & FwFM & DPLR & Pruned & {DPLR vs \par Pruned (\%)} \\
    \midrule
    \multirow[t]{5}{*}{Criteo (AUC)} & 1 & 5.4\% & 0.8044 & 0.8088 & 0.8050 & 0.8008 & 0.52\% & 0.8069 & 0.8100 & 0.8069 & 0.8020 & 0.61\% \\
     & 2 & 10.8\% & 0.8044 & 0.8088 & 0.8066 & 0.8049 & 0.21\% & 0.8069 & 0.8100 & 0.8080 & 0.8057 & 0.29\% \\
     & 3 & 16.2\% & 0.8044 & 0.8088 & 0.8067 & 0.8065 & 0.02\% & 0.8069 & 0.8100 & 0.8083 & 0.8076 & 0.09\% \\
     & 4 & 21.6\% & 0.8044 & 0.8088 & 0.8069 & 0.8074 & -0.05\% & 0.8069 & 0.8100 & 0.8085 & 0.8086 & -0.01\% \\
     & 5 & 27\% & 0.8044 & 0.8088 & 0.8070 & 0.8078 & -0.10\% & 0.8069 & 0.8100 & 0.8085 & 0.8091 & -0.07\% \\
    \midrule
    \multirow[t]{5}{*}{Criteo (LogLoss)} & 1 & 5.4\% & 0.4470 & 0.4429 & 0.4467 & 0.4510 & 0.97\% & 0.4449 & 0.4417 & 0.4449 & 0.4508 & 1.30\% \\
     & 2 & 10.8\% & 0.4470 & 0.4429 & 0.4454 & 0.4469 & 0.34\% & 0.4449 & 0.4417 & 0.4437 & 0.4464 & 0.60\% \\
     & 3 & 16.2\% & 0.4470 & 0.4429 & 0.4450 & 0.4451 & 0.01\% & 0.4449 & 0.4417 & 0.4435 & 0.4444 & 0.21\% \\
     & 4 & 21.6\% & 0.4470 & 0.4429 & 0.4443 & 0.4443 & 0.00\% & 0.4449 & 0.4417 & 0.4431 & 0.4431 & -0.00\% \\
     & 5 & 27\% & 0.4470 & 0.4429 & 0.4443 & 0.4439 & -0.10\% & 0.4449 & 0.4417 & 0.4430 & 0.4426 & -0.08\% \\
    \midrule
    \multirow[t]{5}{*}{Avazu (AUC)} & 1 & 6.4\% & 0.7768 & 0.7777 & 0.7764 & 0.7727 & 0.49\% & 0.7787 & 0.7796 & 0.7776 & 0.7738 & 0.49\% \\
     & 2 & 12.9\% & 0.7768 & 0.7777 & 0.7769 & 0.7756 & 0.16\% & 0.7787 & 0.7796 & 0.7784 & 0.7779 & 0.05\% \\
     & 3 & 19.3\% & 0.7768 & 0.7777 & 0.7772 & 0.7764 & 0.10\% & 0.7787 & 0.7796 & 0.7789 & 0.7789 & 0.00\% \\
     & 4 & 25.8\%  & 0.7768 & 0.7777 & 0.7772 & 0.7770 & 0.03\% & 0.7787 & 0.7796 & 0.7789 & 0.7794 & -0.06\% \\
     & 5 & 32.2\% & 0.7768 & 0.7777 & 0.7774 & 0.7774 & -0.00\% & 0.7787 & 0.7796 & 0.7788 & 0.7796 & -0.10\% \\
    \midrule
    \multirow[t]{5}{*}{Avazu (LogLoss)} & 1 & 6.4\%  & 0.3811 & 0.3808 & 0.3817 & 0.3841 & 0.62\% & 0.3800 & 0.3795 & 0.3810 & 0.3830 & 0.52\% \\
     & 2 & 12.9\% & 0.3811 & 0.3808 & 0.3812 & 0.3817 & 0.11\% & 0.3800 & 0.3795 & 0.3801 & 0.3802 & 0.02\% \\
     & 3 & 19.3\% & 0.3811 & 0.3808 & 0.3812 & 0.3814 & 0.07\% & 0.3800 & 0.3795 & 0.3801 & 0.3800 & -0.02\% \\
     & 4 & 25.8\% & 0.3811 & 0.3808 & 0.3810 & 0.3811 & 0.01\% & 0.3800 & 0.3795 & 0.3800 & 0.3796 & -0.12\% \\
     & 5 & 32.2\% & 0.3811 & 0.3808 & 0.3811 & 0.3810 & -0.05\% & 0.3800 & 0.3795 & 0.3800 & 0.3796 & -0.11\% \\
    \midrule
    \multirow[t]{2}{*}{Movielens (MSE)} & 1 & 33.1\% & 0.7431 & 0.7407 & 0.7449 & 0.7572 & 1.61\% & 0.7376 & 0.7394 & 0.7445 & 0.7630 & 2.43\% \\
     & 2 & 64.3\%  & 0.7431 & 0.7407 & 0.7418 & 0.7464 & 0.62\% & 0.7376 & 0.7394 & 0.7398 & 0.7484 & 1.15\% \\
    \bottomrule
    \end{tabular}
\end{table*}

The first step is computed only once per context in $O(\rho |\ctx| k)$ time, and its results are cached. The last two steps are computed for every item, and their complexity is $O(\rho |\items| k)$ per item, as we desire.

\section{Experiments}\label{sec:experiments}
In this section we present experimental results demonstrating that our approach significantly reduces the cost of item recommendation serving without compromising accuracy, in comparison to the pruning approach. 
%Beyond the results presented here, in \appref{app:posthoc} 
Finally, we demonstrate that finding a DPLR factorization of the field interaction matrix of a regular FwFM, instead of training a DPLR representation, may not be a good approach in practice.

\subsection{Accuracy on public data-sets}
We compare our approach to an FwFM, an FM, and a pruned FwFM using the Criteo Display Advertising Challenge\footnote{https://www.kaggle.com/c/criteo-display-ad-challenge} data-set, comprising of 39 fields and $\sim45\mathrm{M}$ samples, the Avazu challenge dataset\footnote{https://www.kaggle.com/c/avazu-ctr-prediction} comprising of 33 fields and $\sim40\mathrm{M}$ samples, and the Movielens 1M \cite{harper2015movielens} data-set, with 8 fields and $\sim1\mathrm{M}$ samples. Out of the available MovieLens data-sets, we chose the one with the most number of informative context and item fields. The timestamp is converted into year, month, day of week, and hour of day, creating a 11 field data-set. The ``genres'' field in the MovieLens data-set has multiple values, and as described in \secref{sec:inefficiency}, we average the genre embedding vectors to produce one genre vector.

All data-sets were randomly split into 80\% training, 10\% validation, and 10\% test sets. The validation sets were used for tuning the learning rate using Optuna \cite{optuna_2019}. Features with less than 10 occurrences in the training set, and those appearing in the test and validation set that did not appear in the training set, are replaced with a special ``rare feature''. Numerical features are binned similarly to the strategy used by the Criteo Challenge winners \cite{criteo_winners}, via the $x \to \lfloor \ln^2 (x) \rfloor$ function. We tested embedding dimensions of $8$ and $16$. 

We compare DPLR with pruning by matching the number of entries we retain in the pruned field interaction matrix to the number of parameters in the DPLR model. For a rank $\rho$, we use $\rho (m + 1)$ parameters, and the corresponding pruned model is left with the largest magnitude $\rho (m + 1)$ field interaction coefficients. Equivalently, we keep $100 \times \frac{2 \rho (m + 1)}{m (m - 1)}$ percent of the field interactions. 

The results are summarized in \tabref{tbl:public_ds_results}. We see that for aggressive pruning, retaining less than 20\% of the interactions, the DPLR models out-perform, or perform on par, compared the pruned models. For MovieLens, due to a low number of fields, a DPLR model is equivalent to mild pruning, and still outperforms a pruned model.

\Comment{indeed degrades performance,\ariel{it is not clear where the 20\% is coming from, the table is very large}\alex{Added an explicit column to the table} and our DPLR solution with a rank of at least two outperforms the pruned models.\ariel{I don't think 'at least' is appropriate here}\alex{Why? At least and at most are "words" for $\geq$ and $\leq$}}

\subsection{Synthetic latency measurements}
We use a simulation to compare inference time for a large batch of items for various amounts of context fields. We use the Criteo dataset as our example for the number of fields, and simulate vectors from 40 fields. We designate $k \in \{ 10, 15, 20, 25, 30 \}$ of them as ``context'' fields, whereas the rest are item fields. Note, that since Criteo is anonymized, we have no way of knowing which fields are contextual. We measure the scoring time for DPLR models of ranks $\rho \in \{ 1, 2, 3 \}$, and equivalently pruned models that have identical number of parameters by retaining $\rho (m + 1)$ field interactions. The experimental code is implemented in Python, with the score computation parts implemented in Cython to eliminate Python runtime overhead. Moreover, to ensure CPU cache-friendliness of scoring using a pruned model, we fit the 3D array of latent vectors of an entire auction in an appropriate memory format. 

The scoring times for various auction (batch) sizes are reported in \figref{fig:synthetic_latency}, together with scoring times using a regular FwFM. To measure standard error, the average time per auction for each configuration is measured by repeating every scoring experiment 50 times. To ensure precise measurement of each experiment due to timer resolution, each experiment comprises of 10 scored auctions. The experiments are conducted on a 2019 MacBook Pro laptop with a 2.4 GHz 8-Core Intel Core i9 CPU. 

While there are caveats in synthetic timing benchmarks that stem from the implementation and the hardware used, the results allow appreciating the potential of our approach: indeed, scoring using a DPLR model is significantly faster than scoring by a pruned model. Since we made the code available, readers can also appreciate the implementation simplicity of the DPLR scoring algorithm - it relies on batched matrix-matrix products, whose efficient implementations are provided in a variety of numerical computing packages.

\begin{figure*}[h]
    \centering
    \includegraphics[width=.8\linewidth]{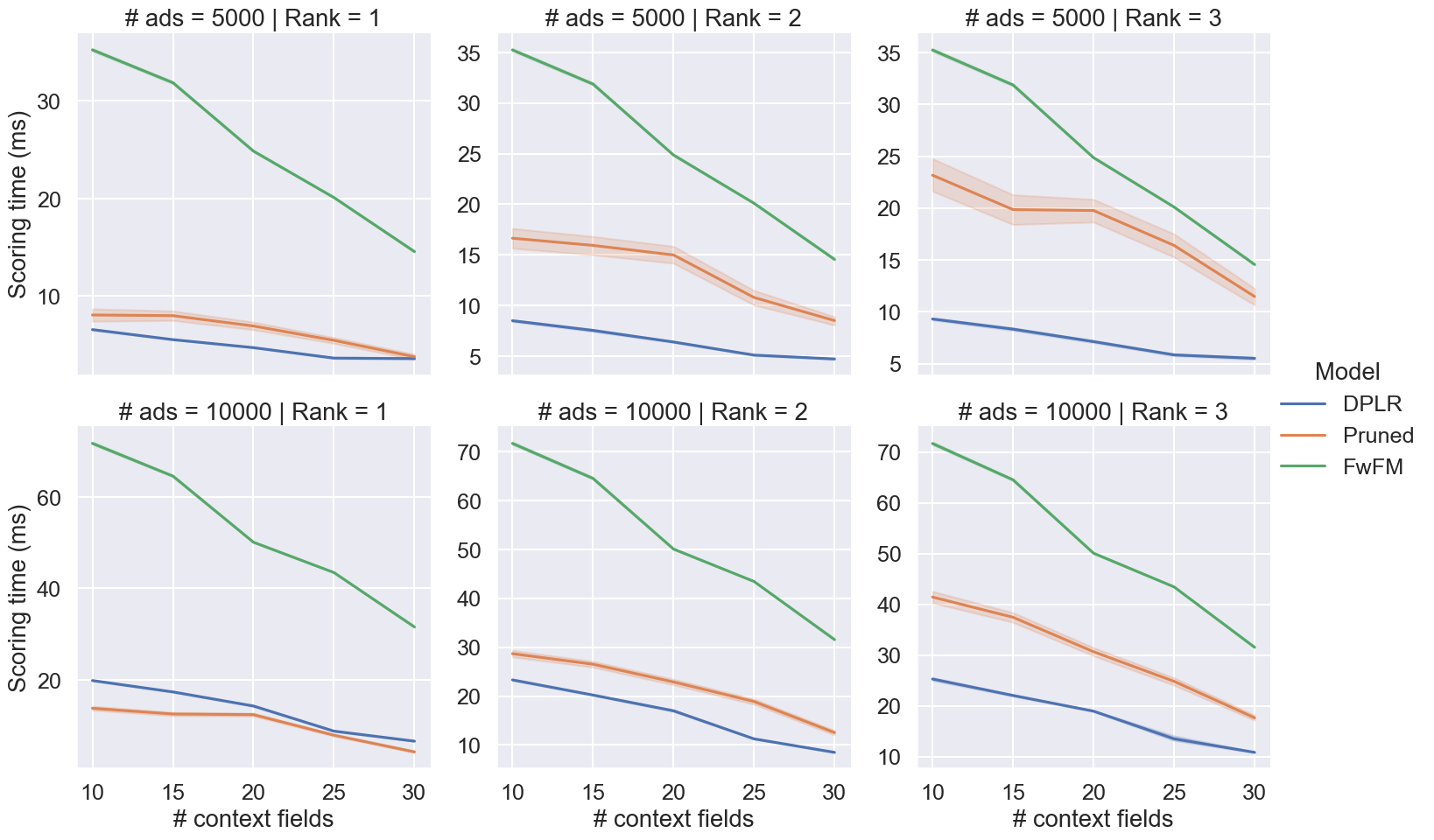}
    \caption{Synthetic timing measurement of a single auction for various auction sizes, DPLR model ranks, and amounts of context fields. Standard error is plotted as a narrow band around each line.}
    \label{fig:synthetic_latency}
\end{figure*}

\subsection{Proprietary online advertising system}
We validate our approach on an online advertising system of a large commercial company, by comparing the trained model's accuracy, and the ad ranking latency. We perform our tests on FwFM models for \emph{click-through rate} (CTR) prediction having 82 fields. To conform to the strict latency requirements, the FwFM's field interaction matrix is pruned to retain only 10\% of the largest magnitude entries.

\subsubsection{Accuracy experiment}
The model trains periodically in a ``sliding window'' mode, meaning that in each period associated with time $T$, the model is trained on logged data from the month before $T$, and is evaluated on logged data from the day before $T$. After sub-sampling to allow training within a reasonable amount of time, the training set comprises tens of millions logged user-ad interactions. We train and evaluate the model on 7 consecutive intervals using the LogLoss and AUC metrics, and average the results by weighting them according to the number of evaluation samples. The trained DPLR model is compared to an FwFM pruned to 10\% of the entries, due to various restrictions of the production code-base. The performance of the models, and the improvement of DPLR over pruned models, are summarized in Table \ref{tbl:dsp_offline}  Surprisingly, the DPLR models perform \emph{better} on the evaluation set than a regular FwFM models. We believe it means that a low-rank field interaction matrix is a good regularization prior for this specific domain. Nevertheless, as our results on public data-sets show, this is not the case in general.

\Comment{
\begin{table*}[h]
    \caption{AUC and LogLoss improvements versus a pruned FwFM retaining 10\% of the interactions on proprietary data. Higher is better.}
    \label{tbl:dsp_offline}
    \centering
    \begin{tabular}{lrrrrrr}
        \toprule
        Metric & \multicolumn{6}{c}{Rank} \\
        \cmidrule(lr){2-7} 
        ~ & 1 & 2 & 3 & 4 & 5 & 6 \\
        \midrule
        LogLoss lift (\%) & -0.478\% & 0.228\% & 0.363\% & 0.409\% & 0.445\% & 0.404\% \\
        AUC lift (\%) & -0.132\% & 0.071\% & 0.098\% & 0.112\% & 0.122\% & 0.115\% \\
    \bottomrule
    \end{tabular}
\end{table*}
}

\begin{table}[h]
    \caption{AUC and LogLoss improvements versus an FwFM on proprietary data. Higher is better.}
    \label{tbl:dsp_offline}
    \centering
    \begin{tabular}{lrrrrrr}
        \toprule
        Metric & \multicolumn{6}{c}{Rank} \\
        \cmidrule(lr){2-7} 
        ~ & 1 & 2 & 3 & 4 & 5 & 6 \\
        \midrule
        LogLoss lift  & -0.48\% & 0.23\% & 0.36\% & 0.41\% & 0.44\% & 0.404\% \\
        AUC lift  & -0.13\% & 0.07\% & 0.10\% & 0.11\% & 0.12\% & 0.115\% \\
    \bottomrule
    \end{tabular}
\end{table}

\subsubsection{Latency experiment}
We deploy the model with rank 3 to a test environment that serves a small portion of the traffic, a few thousand ad ranking queries per minute. The rank was chosen to correspond to the pruning of 90\% of the field interaction entries in terms of the number of parameters. Out of the model's 63 fields, 38 are item fields, and therefore we can theoretically expect approximately $40\%$ latency inference speed improvement in the best case.

We measure the latency incurred by the inference for ad CTR prediction, and the total latency incurred by ranking all eligible ads for a given query. Table \ref{tab:inference_latencies_dsp} summarizes the results, and \appref{app:latency_charts} presents full time-series plots. Evidently, the inference latency is improved by $20\%-30\%$, whereas the query latency by $5\%$, as CTR prediction is only one component of the ad query serving algorithm.

\begin{table}[h]
    \caption{Average and high percentile latency of a single inference or ranking operation in our production system. In each minute, the average, 95th percentile, and the 99th percentile latency for inference, and the 95th percentile latency for ranking. The measurements were averaged over a 10 hour period. The lifts represent improvements - higher is better.}
    \label{tab:inference_latencies_dsp}
    \centering

    \begin{tabular}{lrrrr}
      \toprule
      & \multicolumn{3}{c}{Inference per ad} & Ranking \\
      \cmidrule(lr){2-4}
      \cmidrule(lr){5-5}
      & Average & P95 & P99 & P95 \\
    \toprule
    \Comment{
        Low rank & 9508.92 (15.95) & 24070.06 (22.49) & 32211.04 (60.83) \\
        Pruned & 14465.85 (27.03) & 33952.83 (93.62) & 43276.51 (94.95) \\
        \midrule
    }
    Lift (\%) & 34.27\% & 29.11\% &  25.57\% & 5.45\% \\
    \bottomrule
    \end{tabular}
\end{table}

\subsection{Post-hoc factorization of the field interaction matrix}\label{app:posthoc}
An alternative approach to training a DPLR representation of the field interaction matrix can be training a regular FwFM model, and computing the best DPLR representation we can afterwards. We call this the ``post-hoc'' approach.

Unfortunately, this may not be a good idea. Suppose we obtained an approximation $\tilde{\mR}$ of the model's field interaction matrix $\mR$. Let $\tilde{\mR} = \mR + \mE$, where $\mE$ is the approximation error. Denote by $\lambda_i(\cdot)$ and $\sigma_i(\cdot)$ the $i^\text{th}$ largest eigenvalue and singular value of a given matrix, respectively. By \identref{ident:pairwise_as_trace} we obtain
\begin{align*}
    \sum_{i=1}^m \sum_{j=1}^m \langle \vv_i, \vv_j \rangle \tilde{\emR}_{i,j} 
     &= \Tr(\mV^T \tilde{\mR} \mV) \\
     &= \Tr(\mV^T \mR \mV) + \Tr(\mV^T \mE \mV) \\ 
     &\leq \Tr(\mV^T \mR \mV) + \sum_{i=1}^m \lambda_i(\mV \mV^T) \sigma_i(\mE),
\end{align*}
where the last inequality stems from the Von Neumann trace inequality. Thus, the approximation error boils down to the singular value spectrum of the error matrix. We took the field interaction weights obtained from an FwFM trained on the Criteo data-set, and computed the singular value spectrum of the approximation errors obtained from two approximations: (a) a DPLR approximation of rank 5, computed by minimizing the nuclear norm of the error, and (b) a pruned field interaction matrix where the top 200 entries (same number of parameters as the rank-5 DPLR approximation) were left. The error singular-value spectrum of both approximations is plotted in Figure \ref{fig:error_spectrum}. We observe that the large eigenvalues of the post-hoc DPLR approximation error are \emph{much} larger than the ones of the pruned approximation error. Of course, an upper bound cannot replace a true experimental result on the given dataset, but it gives a rough idea of weather this direction is worth pursuing.

\begin{figure}[h]
    \centering
    \includegraphics[width=.75\columnwidth]{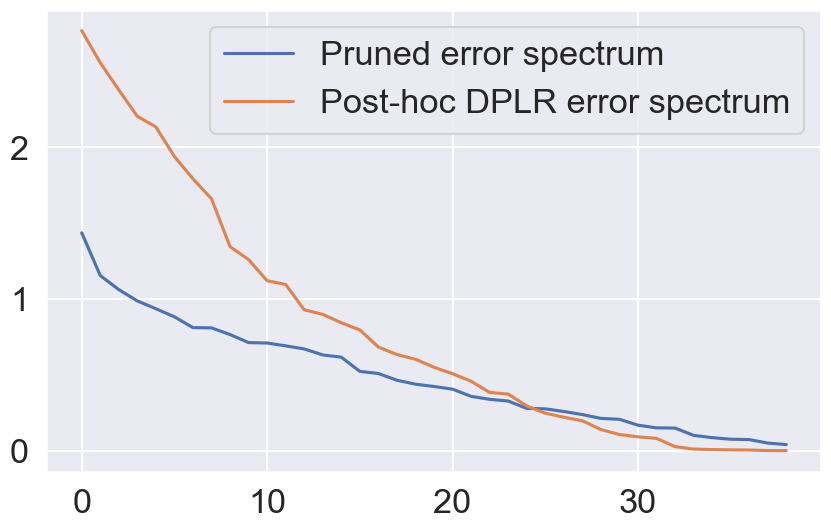}
    \caption{Singular value spectrum of two approximations of the true field interaction matrix. The singular value index is in the horizontal axis, whereas its value is in the verical axis.}
    \label{fig:error_spectrum}
\end{figure}
\section{Conclusions and future work}
In this work we proposed learning a diagonal plus low-rank decomposition of the field interaction matrix of FwFM models as an alternative to the commonly used pruning heuristic in large scale low-latency recommendation systems. We demonstrated that our approach has the potential to outperform pruned models in in terms of both item recommendation speed and model accuracy.

Looking at \eqref{eq:efficient_pairwise}, we observe that the columns of $\mU$ have a notion of ``field importance'' - if all the entries of $\mU_{:, i}$ are negligible, the effect of the $i$-th field is negligible, and it can be discarded. Hence, another future direction is mathematically or experimentally quantifying this notion of field importance, as an alternative to the approach in \cite{kaplan2021dynamic}. This is especially important for real-time systems, where just reducing the number of fields may have a tremendous effect on latency and recommendation costs.
\clearpage
\balance
\bibliographystyle{ACM-Reference-Format}
\bibliography{references}

\clearpage

\appendix
\pagebreak
\section{Latency charts}\label{app:latency_charts}
In \figref{fig:latency_graphs} we observe that the improvement in the latency of the low-rank solution is consistent over time.

\begin{figure}[b]
    \centering
    \begin{subfigure}{0.9\columnwidth}
        \includegraphics[width=\textwidth]{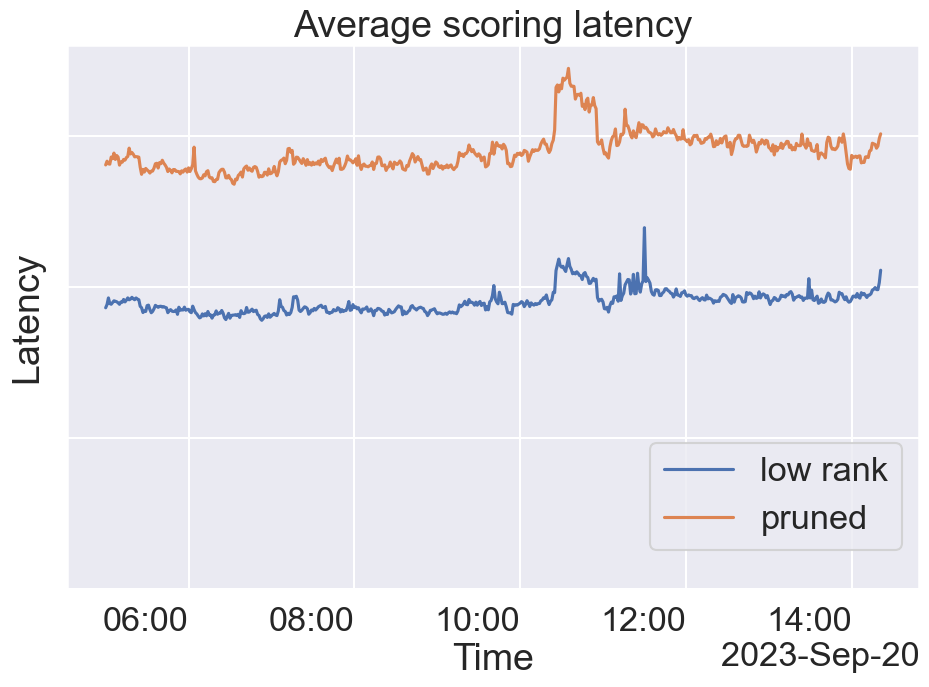}
        \caption{}
    \end{subfigure}
    \hfill
    \begin{subfigure}{0.9\columnwidth}
        \includegraphics[width=\textwidth]{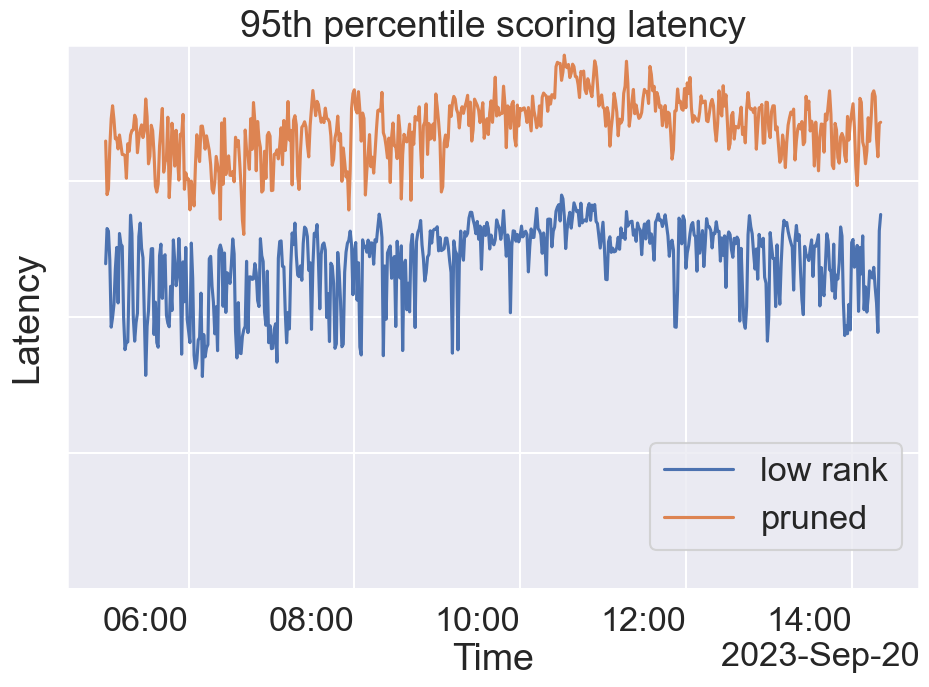}
        \caption{}
    \end{subfigure}
    \hfill
    \begin{subfigure}{0.9\columnwidth}
        \includegraphics[width=\textwidth]{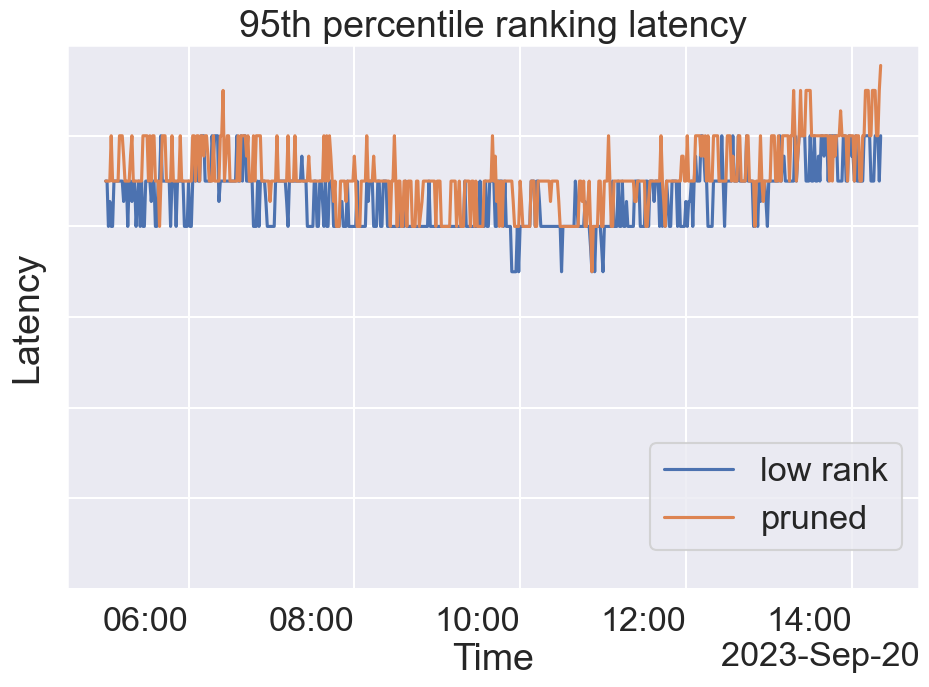}
        \caption{}
    \end{subfigure}
    \caption{Latency graphs over a 10 hour period.}
    \label{fig:latency_graphs}
\end{figure}
\vfill
\end{document}